 \documentclass[onecolumn]{IEEEtran}

  \usepackage{setspace}
  \usepackage{graphicx}
  \usepackage{amsmath}
   \usepackage{amssymb}
   \usepackage{amsthm}
\usepackage{algorithmicx}
\usepackage[Algorithm,ruled]{algorithm}
   \usepackage{algpseudocode}
\floatstyle{boxed}
\usepackage{verbatim}
\restylefloat{figure}
   
\usepackage[latin1]{inputenc}
\usepackage{amsfonts}
\usepackage{fullpage}
\usepackage{hyperref}
\usepackage{todonotes}
\usepackage{cleveref}
\usepackage{breqn}
\usepackage{enumerate}
\setkeys{breqn}{breakdepth={1}}

\theoremstyle{plain}
\newtheorem{theorem}{Theorem}

\newtheorem{lemma}[theorem]{Lemma}

\theoremstyle{definition}
\newtheorem{definition}[theorem]{Definition}
\newtheorem{example}[theorem]{Example}

\title{New Multiple Insertion-Deletion Correcting Codes for Non-Binary Alphabets} 

\author{Tuan A. Le and Hieu D. Nguyen
\thanks{T. Le and H. Nguyen are with the Department of Mathematics, Rowan University, Glassboro, NJ 08012 USA (let5@students.rowan.edu, nguyen@rowan.edu).}
}

\begin{document}


\maketitle

\begin{abstract}

We generalize Helberg's number-theoretic construction of multiple insertion-deletion correcting binary codes to non-binary  alphabets and describe a linear decoding algorithm for correcting multiple deletions.

\end{abstract}

\section{Introduction}

Helberg codes \cite{H} are binary codes capable of correcting multiple insertion-deletion errors.  These number-theoretic codes generalize  Levenshtein codes, first constructed by Varshamov and Tenengo'lts \cite{VT} to correct a single asymmetrical error and later proved by Levenshtein \cite{L} to be capable of also correcting a single insertion or deletion error.  Levenshtein's proof included an elegant linear decoding algorithm to correct a single deletion.  Levenshtein codes are asymptotically optimal; however, Helberg codes correcting more than one insertion or deletion have a low rate \cite{PAFC}.

Other special binary codes capable of correcting insertions and deletions include run-length limited codes by Palun\u{c}i\'{c}, Abdel-Ghaffar, Ferreira, and W. A. Clarke \cite{PAFC}, repetition codes by Landjev and Haralambiev \cite{LH}, and repetition error-correcting codes by Dolecek and Anantharam \cite{DA}.  There are of course codes that can correct insertion-deletion errors with high probability over binary symmetric channels such as concatenated codes by Schulman and Zuckerman \cite{SZ} and watermark codes by Davey and MacKay \cite{DM}.  These codes differ from Helberg and other aforementioned codes, which guarantee correction up to a fixed maximum of insertions and/or deletions (or indels for short).  We refer the reader to \cite{APFC} and \cite{PAFC} for an overview of insertion-deletion correcting codes and their applications.

A non-binary generalization of the Levenshtein code is the Tenengol'ts code \cite{T}, which uses a modular relation to determine the value of the inserted or deleted non-binary symbol and an associated Levenshtein code to determine the position of that symbol.  Tenengolts also gave a systematic form of his code that appends the three-bit string 011 to each codeword to serve as check bits for detecting either an insertion or deletion and as a separator between codewords.  A generalization of the Tenengol'ts code to one capable of correcting multiple indels was constructed by Paluncic, T. G. Swart, J. H. Weber, H. C. Ferreira, and W. A. Clarke \cite{PSWFC}.  As with the Tenengol'ts code, their code uses a set of modular relations to determine the values of the deleted symbols and an associated binary multiple insertion-deletion correcting code to determine the positions of the deleted symbols.  However, this information does not uniquely specify which values should be inserted at these positions; thus, their construction involves a purging process that requires removing unwanted codewords that yield the same deleted codeword.  An upper bound was derived for the number of such codewords that can exist, but no efficient algorithm was given to purge these unwanted codewords.  A lower bound for the cardinality of these codes was established, proving that they are asymptotically optimal, but assumes a conjecture regarding the cardinality of the associated binary code.

In this paper we extend Helberg's construction of his codes \cite{H, HF} to non-binary alphabets.  Moreover, we present a linear decoding algorithm to correct codewords that suffer only deletions.  Our proof that these $q$-ary codes are capable of correcting multiple insertion-deletion errors follows the one given by Abdel-Ghaffar, Palun\u{c}i\'{c}, Ferreira, and Clarke \cite{APFC} for Helberg codes, which we adapt for non-binary alphabets.  The proof relies on an argument by contradiction: suppose two codewords with the same residue produce the same deleted codeword.  Then the difference in their moments must be strictly between 0 and the modulus, which gives a contradiction since the two codewords are congruent.

To precisely describe our results, let $A=\{0,1,...,q-1\}$ be a $q$-ary alphabet and $\mathbf{x}=(x_1,...,x_n)\in A^n$ be a codeword of length $n$.  We shall refer to $x_i$ as the $i$-th symbol of $\mathbf{x}$.  Fix $d$ to be a positive integer and set $p=q-1$.  Generalizing \cite{HF}, we define the sequence of weights $W(q,d)=\{w_1(q,d),w_2(q,d),...\}$ as follows.  First, initialize $w_i(q,d)=0$ if $i\leq 0$.  Then for $i\leq 1$, define $w_i(q,d)$ recursively by
\[
w_i(q,d)=1+p\sum_{j=1}^d w_{i-j}(q,d).
\]
When it is clear, we shall write $w_i$ for short instead of $w_i(q,d)$.  Next, we define the truncated codeword $(\mathbf{x})_k=(x_1,\ldots,x_k)$ to be one consisting of the first $k$ symbols of $\mathbf{x}$ and its moment by $M_k(\mathbf{x}):=M((\mathbf{x})_k)$.
We shall also write $M(x_i)=w_ix_i$ to refer to the moment of the symbol $x_i$.  

Our new $q$-ary codes capable of correcting multiple insertion-deletion errors are defined as follows.

\begin{definition} \label{de:generalized-helberg}
Let $m$ and $r$ be fixed integers satisfying $m\geq w_{n+1}$ and $0\leq r < m$.  We define the code $C_n(q,d,m,r)$ to be the set of codewords of length $n$ whose moments have residue $r$ modulo $m$, i.e.,
\[
C_n:=C_n(q,d,m,r)=\{\mathbf{x}\in A^n: M(\mathbf{x}) \equiv r \ \mathrm{mod} \ m \}.
\] 
\end{definition}
\noindent To simplify the notation, we shall sometimes write $C_n$ instead of $C_n(q,d,m,r)$.  In the case of a binary alphabet where $q=2$, the codes $C_n(2,d,m,r)$ are referred to Helberg codes \cite{HF}.  

Given two codewords $\mathbf{x}$ and $\mathbf{y}$ of length $n$, we shall say that $\mathbf{x}$ and $\mathbf{y}$ are {\em congruent} and write $\mathbf{x}\cong \mathbf{y}$ to denote $M(\mathbf{x}) \equiv M(\mathbf{y}) \ \mathrm{mod} \ m$.  In that case, $\mathbf{x},\mathbf{y}\in C_n(q,d,m,r)$ for some residue $r$ where
\[
r\equiv M(\mathbf{x}) \equiv M(\mathbf{y}) \ \mathrm{mod} \ m.
\]
Moreover, if we define $\Delta(\mathbf{x},\mathbf{y})=M(\mathbf{x}) - M(\mathbf{y})$, then $\mathbf{x}\cong \mathbf{y}$ is equivalent to $\Delta(\mathbf{x},\mathbf{y}) \equiv 0 \ \mathrm{mod} \ m$.

Define $S(n)=\{1,...,n\}$.  Let $D$ be a non-empty subset of $S(n)$ with $|D|\leq d$.  Set $n'=n-|D|$ and define $S'=S(n)-D=\{i_1,...,i_{n'}\}$ with $i_1<i_2<...<i_{n'}$.  Moreover, define $\mathbf{x}^{(D)}=(x_{i_1},...,x_{i_{n'}})$ to be the codeword obtained by deleting the elements of $\mathbf{x}$ indexed by $D$.  We shall refer to $\mathbf{x}^{(D)}$ as a {\em deleted codeword} of $\mathbf{x}$.   We also define the {\em index} of $\mathbf{x}^{(D)}$ to be difference in moments between the original codeword and its deleted codeword:
\[
I:=I(\mathbf{x}^{(D)})=M(\mathbf{x})-M(\mathbf{x}^{(D)}).
\]

We prove in Section \ref{sec:2} that the code $C_n(q,d,m,r)$ is capable of correcting up to $d$ deletion errors.  In particular, let $\mathbf{x}, \mathbf{y} \in C_n(q,d,m,r)$ be two distinct codewords and suppose there exists subsets $D$ and $E$ of $S(n)$ such that $|D|=|E|\leq d$ and $x^{(D)}=y^{(E)}$.  We show that $0< \Delta(\mathbf{x},\mathbf{y})<m$, which is a contradiction since $\mathbf{x}\cong \mathbf{y}$.  Thus, no such subsets exist.  Hence, $C_n(q,d,m,r)$ is a $d$-deletion correcting code.  By a result of Levenshtein \cite{L}, $C_n(q,d,m,r)$ is also capable of correcting a total of $d$ indels.

In Section \ref{sec:3}, we present a linear search algorithm to decode codewords in $C_n$ that suffer only deletions.  Suppose a codeword $\mathbf{x}$ is transmitted, but is corrupted so that the received codeword, denoted by $\mathbf{x}'$, consists of deletion errors.  The goal of our algorithm to find the correct positions to re-insert into $\mathbf{x}'$ the symbols that were deleted so that the index $I$ reduces to zero.  In particular, we start with the assumption that our deleted symbols should be inserted at the right end of $\mathbf{x}'$.  If these symbols are not in their correct positions, then we shift them to the left as far as possible and update the index $I$ by subtracting the change in the $weight$ of each moving symbol from the current value of $I$.  The algorithm terminates when $I=0$.  For the correction of one deletion error, the algorithm essentially performs an exhaustive trial-by-error search.  However, for two or more deletion errors, the algorithm is recursive in the following sense:  assuming that $d$-deletion errors have occurred, the algorithm corrects the rightmost deleted bit, after which the decoding reduces to the algorithm for correcting $(d-1)$-deletion errors.  Moreover, for $d\geq 2$, the algorithm is efficient because its complexity is linear, namely $O(n)$, where $n$ is the length of the transmitted codeword.

Lastly, in the Appendix B, we present values for the size of the largest code $C_n(q,d,m,r)$ for certain values of $q$, $d$, and $n$.  These values were found through exhaustive computer search.

\section{Generalized Helberg Codes} 
\label{sec:2}

Our proof that $C_n(q,d,m,r)$ is a $d$-deletion error-correcting code follows the proof given in \cite{APFC}, where we adapt their arguments for $q$-ary alphabets.
We shall need the following lemma, which allows us to replace the rightmost non-zero bit with the value 0 in any two codewords that are congruent and have the same deleted codeword.  This assumes that the rightmost nonzero bit is the same for both codewords.

\begin{lemma} \label{le:replace-bit}
Let $\mathbf{x}$ and $\mathbf{y}$ be two codewords of length $n$ with the following two properties:
\begin{enumerate}[(1)]
\item $\mathbf{x} \cong \mathbf{y}$.

\item $\mathbf{x}^{(D)}=\mathbf{y}^{(E)}$ for some subsets $D$ and $E$ of $\{1,...,n\}$ with $|D|=|E|\leq d$.
\end{enumerate}

\noindent Suppose there exists a positive integer $L$ such that $x_{L}=y_{L}>0$ and $x_i=y_i=0$ for all $i>L$.  
Then  there exist codewords $\tilde{\mathbf{x}}$ and $\tilde{\mathbf{y}}$ where $\tilde{x}_i=x_i$, $\tilde{y}_i=y_i$ for all $i\neq L$ and $\tilde{x}_{L}=\tilde{y}_{L}=0$ such that $\tilde{\mathbf{x}}$ and $\tilde{\mathbf{y}}$ have the same two properties as $\mathbf{x}$ and $\mathbf{y}$, namely
\begin{enumerate}[(i)]
\item $\tilde{\mathbf{x}} \cong \tilde{\mathbf{y}}$.

\item $\tilde{\mathbf{x}}^{(\tilde{D})}=\tilde{\mathbf{y}}^{(\mathbf{\tilde{E}})}$ for some sets $\tilde{D}$ and $\tilde{E}$ having the same size as $D$ and $E$. 

\end{enumerate}
\end{lemma}

\begin{proof} Define $\tilde{\mathbf{x}}$ and $\tilde{\mathbf{y}}$ according to the lemma.  Since $x_i-y_i=\tilde{x}_i-\tilde{y}_i$ for all $i=1,...,n$, it follows that $M(\mathbf{x}) - M(\mathbf{y})=M(\tilde{\mathbf{x}})-M(\tilde{\mathbf{y}})$.  But $\mathbf{x} \cong \mathbf{y}$; hence, $\tilde{\mathbf{x}}\cong \tilde{\mathbf{y}}$.  This proves (i).  To prove (ii), we consider four cases:

\vskip 6pt
\noindent Case I: Assume $L \in D\cap E$.  In this case, the nonzero bits $x_{L}$ and $y_{L}$are deleted from $\mathbf{x}$ and $\mathbf{y}$, respectively, to obtain $\mathbf{x}^{(D)}$ and $\mathbf{y}^{(E)}$.  Define $\tilde{D}=D$ and $\tilde{E}=E$.  Since $\mathbf{x}^{(D)}=\mathbf{y}^{(E)}$, it follows that $\tilde{\mathbf{x}}^{(\tilde{D})}=\tilde{\mathbf{y}}^{(\tilde{E})}$ since the zero bits $\tilde{x}_{L}$ and $\tilde{y}_{L}$ are deleted from $\tilde{\mathbf{x}}$ and $\tilde{\mathbf{y}}$, respectively, as well.

\vskip 6pt
\noindent Case II: Assume $L\not\in D\cup E$.  Since $\mathbf{x}^{(D)}=\mathbf{y}^{(E)}$, it follows that $x_{L}$ and $y_{L}$ appear in $\mathbf{x}^{(D)}$ and $\mathbf{y}^{(E)}$ as the rightmost nonzero bit, respectively.  But then replacing $x_{L}$ and $y_{L}$ by $\tilde{x}_{L}$ and $\tilde{y}_{L}$, respectively, yields $\tilde{\mathbf{x}}^{(D)}=\tilde{\mathbf{y}}^{(E)}$.  Thus, it suffices to again define $\tilde{D}=D$ and $\tilde{E}=E$.  

\vskip 6pt
\noindent Case III: Assume $L \in D-E$.  In this case, the bit $x_{L}$ is deleted from $\mathbf{x}$ to obtain $\mathbf{x}^{(D)}$, but the bit $y_{L}$ is not deleted from $\mathbf{y}$ and therefore appears in $\mathbf{y}^{(E)}$.  Let $z$ denote the number of bits to the right of $y_{L}$ in $\mathbf{y}^{(E)}$, which must all be 0 since $y_i=0$ for all $i>L$.  Then the number of bits to the right of $y_{L}$ that are deleted from $\mathbf{y}$ to obtain $\tilde{\mathbf{y}}^{(E)}$ equals $z'=n-L-z$.  Let $K$ denote the position of the rightmost nonzero bit $x_K$ of $\mathbf{x}^{(D)}$.  Since $\mathbf{x}^{(D)}=\mathbf{y}^{(E)}$, it follows that $x_K=y_{L}=x_{L}$ and that the number of zeros to the right of $x_K$ in $\mathbf{x}^{(D)}$ also equals $z$.  Therefore, the number of bits to the right of $x_K$ that are deleted from $\mathbf{x}$ to obtain $\mathbf{x}^{(D)}$ equals $n-K-z$. We now define $D'=\{K,K+1,...,L-1,L+1,...,L+z'\}$ where we exclude $L$.  It follows that $\mathbf{x}^{(D')}=\mathbf{x}^{(D)}$ with $|D'|=|D|$.  Since $L \not\in D'\cup E$, this reduces to Case II where $D$ is replaced by $D'$.

\vskip 6pt
\noindent Case IV: Assume $L\in E-D$.  The argument in this case is the same as Case III with the roles of $D$ and $E$ reversed.
\end{proof}

\begin{theorem}\label{th:inequality}
Let $\mathbf{x}$ and $\mathbf{y}$ be two codewords of length $n$ that satisfy properties (1) and (2) in Lemma \ref{le:replace-bit}.
Then
\[
0<|\Delta(\mathbf{x},\mathbf{y})|<m.
\]
\end{theorem}

\begin{proof} We shall first prove that $|\Delta(\mathbf{x},\mathbf{y})|<m$. To begin, we rewrite $\Delta(\mathbf{x},\mathbf{y})$ as follows:
\begin{align*}
\Delta(\mathbf{x},\mathbf{y}) 
& =  M(\mathbf{x}) - M(\mathbf{y}) \\
& = \sum_{i\in D}w_ix_i - \sum_{j\in E} w_j y_j +\sum_{k=1}^{n'}(w_{i_k} - w_{j_k})x_{i_k},
\end{align*}
where $n'=n-|D|$.
This yields the bound
\begin{align*}
\Delta(\mathbf{x},\mathbf{y}) & \leq \sum_{i\in D}w_ix_i +\sum_{k=1}^{n'}(w_{i_k} - w_{j_k})x_{i_k}.
\end{align*}
Next, we partition $S(n')=\{1,2,...,n'\}$ into those elements $k$ where $i_k\leq j_k$ and those where $i_k>j_k$ to obtain
\begin{align*}
\Delta(\mathbf{x},\mathbf{y})  & \leq  \sum_{i\in D}w_ix_i +\sum_{\substack{k\in S(n') \\ i_k \leq j_k }}(w_{i_k} - w_{j_k})x_{i_k} \\
& \ \ \ \ + \sum_{\substack{k\in S(n') \\  i_k > j_k }}(w_{i_k} - w_{j_k})x_{i_k} \\
& \leq  \sum_{i\in D}w_ix_i + \sum_{\substack{k\in S(n') \\  i_k > j_k }}(w_{i_k} - w_{j_k})x_{i_k} \\
& \leq  \sum_{i\in D}p w_i + \sum_{\substack{k\in S(n') \\  i_k > j_k }}p (w_{i_k} - w_{j_k}) \\
& =  \sum_{i\in D}p w_i + \sum_{\substack{k\in S(n') \\  i_k > j_k }}p w_{i_k} - \sum_{\substack{k\in S(n') \\  i_k > j_k }}p w_{j_k}.
\end{align*}
We now add and subtract as follows:
\begin{align*}
\Delta(\mathbf{x},\mathbf{y}) & \leq  \sum_{i\in D}p w_i + \sum_{\substack{k\in S(n') \\  i_k > j_k }}p w_{i_k} +\sum_{\substack{k\in S(n') \\  i_k \leq j_k }}p w_{i_k} \\
& \ \ \ \ -\sum_{\substack{k\in S(n') \\  i_k \leq j_k }}p w_{i_k} - \sum_{\substack{k\in S(n') \\  i_k > j_k }}p w_{j_k} \\
& =  \sum_{i=1}^n p w_i -\sum_{k=1}^{n'} p w_{\min(w_{i_k},w_{j_k})} \\
& \leq  \sum_{i=1}^n p w_i -\sum_{k=1}^{n'} p w_{k} \\
& \leq  \sum_{i=n'+1}^{n} p w_i =  \sum_{j=1}^{n-n'} p w_{n+1-j}   \\
& \leq  p\sum_{j=1}^{d} w_{n+1-j} =w_{n+1}-1 \\
& < m.
\end{align*}
On the other hand, by reversing the roles of $\mathbf{x}$ and $\mathbf{y}$, we obtain $\Delta(\mathbf{y},\mathbf{x}) <m$, which implies $\Delta(\mathbf{x},\mathbf{y})=-\Delta(\mathbf{y},\mathbf{x})>-m$.  Hence,  $|\Delta(\mathbf{x},\mathbf{y})|<m$ as desired.

Next, we prove that $\Delta(\mathbf{x},\mathbf{y})\neq 0$ by considering four different cases.  By Lemma 1 we can assume without loss of generality that there exists an integer $L\in \{1,...,n\}$ such that
$x_L>y_L$ and $x_i=y_i=0$ for all $i>L$.

\vskip 6pt

\noindent Case I: Assume $L\in D\cap E$.  Then $i_k \neq L$ for all $k=1,...,n'$. Decompose
\begin{align*}
\Delta(\mathbf{x},\mathbf{y}) & = \sum_{i\in D}w_ix_i - \sum_{j\in E} w_j y_j +\sum_{k=1}^{n'}(w_{i_k} - w_{j_k})x_{i_k}. \\
\end{align*}
The first two summations on the right-hand side is bounded below by
\begin{flalign*}
& \sum_{i\in D}w_ix_i - \sum_{j\in E} w_j y_j \\
&  = w_Lx_L +\sum_{\substack {i \in D \\ i \leq L-1}} w_i x_i 
 - w_Ly_L -  \sum_{\substack {j \in E \\ j\leq L-1}} w_j y_j \\
&  \geq w_L-  \sum_{\substack {j \in E \\ j\leq L-1}} w_j y_j
\end{flalign*}
The third summation is bounded below by
\begin{flalign*}
&\sum_{k=1}^{n'}(w_{i_k} - w_{j_k})x_{i_k} \\
& =\sum_{\substack{k\in S(n') \\ i_k < j_k }}(w_{i_k} - w_{j_k})x_{i_k} + \sum_{\substack{k\in S(n') \\  i_k \geq j_k }}(w_{i_k} - w_{j_k})x_{i_k} \\
& \geq \sum_{\substack{k\in S(n') \\ i_k < j_k, i_k \leq L-1 }}(w_{i_k} - w_{j_k})x_{i_k},
\end{flalign*}
where we have used the fact that $i_k\neq L$ and $x_{i_k}=0$ for $i_k >L$.
It follows that
\[
\begin{aligned}
& \Delta(\mathbf{x},\mathbf{y}) \\
& \geq w_L -  \sum_{\substack {j \in E \\ j\leq L-1}} w_j y_j  +\sum_{\substack{k\in S(n') \\ i_k < j_k, i_k \leq L-1 }}p(w_{i_k} - w_{j_k})
\end{aligned}
\]
since  $x_{i_k} \leq p$.
Next, we use the decomposition
\begin{align}
\sum_{j=1}^{L-1} p w_j & = \sum_{\substack {j \in E \\ j\leq L-1}} p w_j + \sum_{\substack{k\in S(n') \\ i_k <  j_k, i_k \leq L-1 }}pw_{j_k} \label{eq:decomposition} \\
& \ \ \ \ + \sum_{\substack{k\in S(n') \\ i_k \geq j_k, i_k \leq L-1 }}pw_{j_k}  \notag
\end{align}
to obtain
\begin{align*}
\Delta(\mathbf{x},\mathbf{y}) 
& \geq w_L -  \sum_{j=1}^{L-1} p w_j +\sum_{\substack{k\in S(n') \\ i_k < j_k, i_k \leq L-1 }}pw_{i_k} \\
& \ \ \ \ +\sum_{\substack{k\in S(n') \\ i_k \geq j_k, i_k \leq L-1 }}pw_{j_k}.
\end{align*}
This equivalent to
\begin{align*}
\Delta(\mathbf{x},\mathbf{y}) 
& \geq w_L -  \sum_{j=1}^{L-1} p w_j +\sum_{\substack{k\in S(n') \\ i_k \leq L-1 }}pw_{\mathrm{min}(i_k,j_k)}.
\end{align*}
Since  $k\leq \mathrm{min}(i_k,j_k)$, we have
\begin{align*}
\Delta(\mathbf{x},\mathbf{y}) & \geq w_L -  \sum_{j=1}^{L-1} p w_j +\sum_{k=1}^{\mathrm{min}(n',L-1)}pw_k \\
& = w_L -  \sum_{i=\mathrm{min}(n',L-1)}^{L-1} p w_i \\
& \geq w_L -  \sum_{i=L-d}^{L-1} p w_i \\
& \geq 1,
\end{align*}
where we have used the fact that $L-d \leq \mathrm{min}(n',L-1)$.  Also, recall that $L \leq n=n'+d$ and $d\geq 1$.

\noindent Case II: Assume $L\in D-E$.  Recall that $x_L>y_L$ and $x_i=y_i=0$ for $i>L$.  Since $L\notin E$, it follows that
\begin{align*}
\Delta(\mathbf{x},\mathbf{y})  
& = \sum_{i\in D}w_ix_i - \sum_{j\in E} w_j y_j +\sum_{k=1}^{n'}(w_{i_k} - w_{j_k})x_{i_k} \\
& = w_Lx_L +\sum_{\substack {i \in D \\ i \leq L-1}} w_i x_i  -  \sum_{\substack {j \in E \\ j\leq L-1}} w_j y_j \\
& \ \ \ \ +\sum_{k=1}^{n'}(w_{i_k} - w_{j_k})x_{i_k}.
\end{align*}
Analogously, we partition $S(n')$ into those elements $k$ where $i_k<j_k$ and those where $i_k\geq j_k$ to obtain
\begin{align*}
& \Delta(\mathbf{x},\mathbf{y})  \\
& \geq w_Lx_L -  \sum_{\substack {j \in E \\ j\leq L-1}} w_j y_j +\sum_{\substack{k\in S(n') \\ i_k < j_k }}(w_{i_k} - w_{j_k})x_{i_k} \\
& \ \ \ \ + \sum_{\substack{k\in S(n') \\  i_k \geq j_k }}(w_{i_k} - w_{j_k})x_{i_k} \\
& \geq w_Lx_L -  \sum_{\substack {j \in E \\ j\leq L-1}} w_j y_j +\sum_{\substack{k\in S(n') \\ i_k < j_k }}(w_{i_k} - w_{j_k})x_{i_k} \\
& \geq w_L -  \sum_{\substack {j \in E \\ j\leq L-1}} w_j y_j \\
& \ \ \ \ +\sum_{\substack{k\in S(n') \\ i_k < j_k, i_k \leq L-1 }}(w_{i_k} - w_{j_k})x_{i_k}  
\end{align*}
The rest of the argument now follows the same as that in Case I to establish that $\Delta(\mathbf{x},\mathbf{y})\geq 1$.

\vskip 6pt
\noindent Case III. Assume $L \in E-D$.  The argument in this case is the same as Case II by switching the roles of $D$ and $E$.

\vskip 6pt
\noindent Case IV. Assume $L \not\in D\cup E$.  Then $i_K=L$ for some $i_K \in S'$.  We claim that $j_K\leq i_K -1$.  Since $\mathbf{x}^{(D)}=\mathbf{y}^{(E)}$, it follows that $x_{i_K}=y_{j_K}$.  On the other hand, we have $y_{i_K}<x_{i_K}$ and $y_i=0$ for all $i\geq L=i_K$.  Thus, $j_K\leq i_K-1$. 

We now proceed similarly as in previous cases:
\[
\begin{aligned}
& \Delta(\mathbf{x},\mathbf{y})  \\
& = \sum_{\substack {i \in D \\ i \leq L-1}} w_i x_i  -  \sum_{\substack {j \in E \\ j\leq L-1}} w_j y_j +\sum_{k=1}^{n'}(w_{i_k} - w_{j_k})x_{i_k} \\
& \geq -  \sum_{\substack {j \in E \\ j\leq L-1}} w_j y_j +\sum_{\substack{k\in S(n') \\ i_k < j_k }}(w_{i_k} - w_{j_k})x_{i_k} \\
& \ \ \ \ + \sum_{\substack{k\in S(n') \\  i_k \geq j_k }}(w_{i_k} - w_{j_k})x_{i_k} \\
&  \geq -  \sum_{\substack {j \in E \\ j\leq L-1}} w_j y_j + \sum_{\substack{k\in S(n') \\ i_k < j_k, i_k \leq L-1 }}(w_{i_k} - w_{j_k})x_{i_k} \\
& \ \ \ \ + (w_{i_K}-w_{j_K})x_{i_K}.
\end{aligned}
\]
Next, since $x_i \leq p$ for all $i\in \mathbb{N}$, we have
\[
\begin{aligned}
& \Delta(\mathbf{x},\mathbf{y}) \\
& =  w_L - pw_{j_K} -  \sum_{\substack {j \in E \\ j\leq L-1}} p w_j \\
& \ \ \ \ +\sum_{\substack{k\in S(n') \\ i_k < j_k, i_k \leq L-1 }}pw_{i_k} - \sum_{\substack{k\in S(n') \\ i_k < j_k, i_k \leq L-1 }}pw_{j_k}. \\
\end{aligned}
\]
Again, using (\ref{eq:decomposition}), we obtain the lower bound
\begin{align*}
\Delta(\mathbf{x},\mathbf{y}) & \geq w_L -  \sum_{j=1}^{L-1} p w_j +\sum_{\substack{k\in S(n') \\ i_k < j_k, i_k \leq L-1 }}pw_{i_k} \\
& \ \ \ \ +\sum_{\substack{k\in S(n') \\ i_k \geq j_k, i_k \leq L-1 }}pw_{j_k}. \\
\end{align*}
The rest of the proof  is the same as that in Case I.  Therefore, $\Delta(\mathbf{x},\mathbf{y})\geq 1$.  Hence, $0<|\Delta(\mathbf{x},\mathbf{y})|<m$ as desired.
\end{proof}

\begin{theorem}
The code $C_n(q,d,m,r)$ is a $d$-insertion-deletion correcting code.
\end{theorem}

\begin{proof}
Suppose on the contrary that $C_n(q,d,m,r)$ is not capable of correcting up to $d$ deletions.  Then there exist two codewords $\mathbf{x},\mathbf{y}\in C_n(q,d,m,r)$ and subsets $D$ and $E$ with $|D|=|E|\leq d$ such that $\mathbf{x}^{(D)} = \mathbf{y}^{(E)}$.  By Theorem \ref{th:inequality}, we have $0< |\Delta(\mathbf{x},\mathbf{y})|<m$.  It follows that $\mathbf{x}\not\cong \mathbf{y}$, a contradiction.  Thus, $C_n(q,d,m,r)$ is capable of correcting up to $d$ deletions, and therefore, can correct up to $d$ insertion-deletion errors as well due to a result of Levenshtein \cite{L}.
\end{proof}

\section{Decoding of Generalized Helberg codes}
\label{sec:3}

In this section, we describe a linear decoding algorithm to correct deletion errors in a generalized Helberg codeword $\mathbf{x} \in C_n(2,d,m,r)$ where $c$ deletions have occurred with $c\leq d$.  We first present an algorithm to correct one deletion and then provide a recursive algorithm to correct two or more deletions.

\subsection{Decoding One Deletion}
In the decoding of one deletion, our algorithm is the same as exhaustive trial-by-error search. Let $\mathbf{x} \in C_n(q,d,m,r)$ and $\mathbf{x}'$ be the deleted codeword obtain from $\mathbf{x}$ by deleting one symbol.  We assume $d \geq 2$; otherwise, if $d=1$, Levenshtein decoding should be used. Then $\mathbf{x}'$ has length $n-1$.  We define $\mathbf{\tilde{x}}=(\tilde{x}_1,\ldots,\tilde{x}_n)$ to be the initial decoding of $\mathbf{x}'$ where we append a variable symbol $\delta$ to $\mathbf{x}'$ at initial position $P=n$, i.e., the right-most position:
\[
\tilde{\mathbf{x}}=(x'_1,x'_2,\ldots, x'_{n-1},\delta).
\]
Let $I=M(\mathbf{x})-M(\mathbf{x}')$ denote the index (Lemma \ref{le:appendix-3} in Appendix A shows that it possible to determine $M(\mathbf{x})$ from $M(\mathbf{x}')$).  We then attempt to decode $\tilde{\mathbf{x}}$ in order to obtain the original codeword $\mathbf{x}$ so that $M(\tilde{\mathbf{x}})=M(\mathbf{x})$ by either inserting a value for $\delta$ or shifting this deleted symbol to the left of $x'_{P-1}$.  The decision is based on the following condition, which compares the current index $I$ and the moment of $\delta$ at position $P$: \\
\\
\textbf{Algorithm D1 (Decode One Deletion)}: Let $P=n$.  If $I=\sigma \cdot w_P$ for some value $\sigma \in \{0,1,\ldots, p\}$, then $\delta$ is in its correct position as the symbol that was deleted from $\mathbf{x}$.  To decode, set $\delta=\sigma$. 
Otherwise, shift $x'_{P-1}$ to the right of $\delta$ (equivalent to shifting $\delta$ to the left one position), update $I\rightarrow I-x'_{P-1}(w_P-w_{P-1})$, and update $P\rightarrow P-1$.  This is repeated until the the correct position and value for $\delta$ is found. \\

It is clear that algorithm D1 will correctly decode $\mathbf{x}'$ since it essentially performs an exhaustive search (assuming that $\mathbf{x}$ exists).  We illustrate this algorithm in the following example.

\begin{algorithm}[H]
 \floatname{algorithm}{Algorithm D1}
 \caption{}
\begin{algorithmic}[1]
\State $\tilde{\mathbf{x}}=x'_1  x'_2 x'_3 ...x'_{n-1}\delta$\indent\Comment{Initialize $\tilde{\mathbf{x}}$ by appending variable symbol $\delta$ to $\mathbf{x}'$ at position $n$, where $\delta$ is to be determined.}
\For{$P=n$ to $1$} \Comment{$P$ denotes position of $\delta$}
\For{$\sigma=q-1$ to $0$} \Comment{$\sigma$ denotes test value for $\delta$}
\If{$I=\sigma \cdot w_P$}
\State $\delta=\sigma$
\State $STOP$
\EndIf
\EndFor
\State $\tilde{\mathbf{x}}= x'_1  x'_2 x'_3 \ldots x'_{P-2} \delta x'_{P-1} x'_{P+1} \ldots x'_{n-1}$ \Comment{Shift $x'_{P-1}$ to the right of $\delta$ and update $\tilde{\mathbf{x}}$}
\State $I=I-x'_{P-1}(w_P-w_{P-1})$ \Comment{Update the index}
\EndFor
	
\end{algorithmic}
\caption{(Decode One Deletion)}
\end{algorithm}

\begin{example} \label{ex:one-deletion}
Suppose the ternary codeword $\mathbf{x}=(1,2,2,0,2,2,1,2) \in C_8(3,2,w_9,23)$ was transmitted and $\mathbf{x}'=(1,2,2,0,2,1,2)$ was received so that one deletion occurred.  We wish to decode $\mathbf{x}'$ to recover $\mathbf{x}$.   The weights $w_i$ corresponding to this codebook are defined by the recursion $w_i=1+2(w_{i-1}+w_{i-2})$.  The first 10 weights are given in Table \ref{table:weights}.  In particular, $w_9=3861$.

 \begin{table}[ht]
\caption{Weights $w_i$ for $d=2,\;q=3$}
\centering
\begin{tabular}{||c| c|| c| c||}
\hline
$i$ & $w_i$ & $i$ & $w_i$ \\
\hline
1 & 1 & 6 & 189 \\
\hline
2 & 3 & 7 & 517 \\
\hline
3 & 9 & 8 & 1413 \\
\hline
4 & 25 & 9 & 3861 \\
\hline
5 & 69 & 10 & 10549 \\
\hline
\end{tabular}
\label{table:weights}
\end{table}

Since $m=w_9=3861$, $r=23$, and $M(\mathbf{x}')=1386$, and $M(\mathbf{x}') > r$, it follows from Lemma \ref{le:appendix-3} in Appendix A that $M(\mathbf{x})=3884$. Thus, the index $I=M(\mathbf{x})-M(\mathbf{x}')=2498$.  As defined earlier, let $\tilde{\mathbf{x}}$ be our initial decoding for $\mathbf{x}'$ where we initially insert a variable symbol $\delta$ at the right-most position of $\mathbf{x}'$, namely at position $P=8$:
\[
\tilde{\mathbf{x}}=(1,2,2,0,2,1,2,\delta ).
\]
According to algorithm D1, since $I\neq \sigma \cdot w_8$ for all $\sigma \in \{0,1,2\}$, we shift $\tilde{x}_7=2$ to the right of $\delta$, update $P\rightarrow P-1=7$, and update the index $I\rightarrow I-x_7(w_8-w_7)=706$ so that $\tilde{\mathbf{x}}$ now appears as
\[
\tilde{\mathbf{x}}=(1,2,2,0,2,1,\delta,2).
\]
Again, since $I\neq \sigma \cdot w_7$ for all $\sigma\in \{0,1,2\}$, we shift $\tilde{x}_6=1$ to the right of $\delta$, update $P\rightarrow P-1=6$, and update $I\rightarrow I-x_6(w_7-w_6)=378$.  Then
\[
\tilde{\mathbf{x}}=(1,2,2,0,2,\delta,1,2).
\]
We now find that $I=\sigma*w_6=378$ for $\sigma=2$.  In that case, we set $\delta=2$ and set $I=0$. This gives the original codeword
\[
\tilde{\mathbf{x}}=(1,2,2,0,2,2,1,2)=\mathbf{x}
\]
and completes the decoding.
\end{example}

\subsection{Decoding Two Deletions}
For binary Helberg codes capable of correcting two deletions, we shall describe a recursive algorithm to decode a codeword where two symbols have been deleted by reducing the problem to that of correcting one deletion, a problem that was solved in the previous sub-section. 

Suppose $\mathbf{x}'$ is obtained from $\mathbf{x} \in C_n(2,2,m,r)$ after deleting two symbols from $\mathbf{x}$. Then to decode $\mathbf{x}'$, whose length is $n-2$, we again define $\tilde{\mathbf{x}}$ to be an initial decoding of $\mathbf{x}$ where we insert two variable symbols $\delta_1$ and $\delta_2$ at the right end of $\mathbf{x}'$, namely at positions $P-1$ and $P$, where we initially set $P=n$:
\[
\tilde{\mathbf{x}}=(x'_1,x'_2,\ldots,x'_{n-2},\delta_1,\delta_2).
\]
We calculate $I=M(\mathbf{x})-M(\mathbf{x}')$ (use Lemma \ref{le:appendix-3} in Appendix A to determine $M(\mathbf{x})$).  Our algorithm essentially determines whether to set $\delta_2$ equal to an alphabet symbol (0 or 1), in which case the decoding reduces to the one-deletion algorithm D1, or shift $x'_{P-2}$ (initially $x'_{n-1}$) to the right of $\delta_2$.  The following conditions describe when each action is executed.  \\
\\
\noindent{\bf Algorithm D2-Binary (Decode Two Deletions):} Let $P=n$. If 
\begin{equation} \label{eq:test}
I=\sigma_1 w_{P-1}+\sigma_2 w_P
\end{equation}
for some $\sigma_1,\sigma_2 \in \{0,1\}$, then $\delta_1$ and $\delta_2$ are in their correct positions as symbols that were deleted from $\mathbf{x}$.  To decode, set $\delta_1=\sigma_1$ and $\delta_2=\sigma_2$. 

Otherwise, we assume that either $\delta_1$ or $\delta_2$ (or both) are NOT in their correct positions in what follows.  Then
\begin{enumerate}
\item For $w_P>I$: 
\begin{enumerate}
\item If $x'_{P-2}=0$, then shift $x'_{P-2}$ to the right of $\delta_2$, i.e., to the right of $\delta_2$.
\item If $x'_{P-2}=1$ and 
\begin{enumerate}
\item $I<w_P-w_{P-2}$, then set $\delta_2=0$ and update the index $I \rightarrow I-(w_{P-1}-w_{P-2})$.
\item $I\geq w_P-w_{P-2}$, then shift $x'_{P-2}$ to the right of $\delta_2$ and update the index $I \rightarrow I-(w_P-w_{P-2})$.
\end{enumerate}
\end{enumerate}
\item For $w_P < I$: 
\begin{enumerate}
\item If $x'_{P-2}=0$, then set $\delta_2=1$ and update $I \rightarrow I-w_P$.
\item If $x'_{P-2}=1$, then shift $x'_{n-2}$ to the right of $\delta_2$ and update $I \rightarrow I-(w_P-w_{P-2})$. 
\end{enumerate}
\end{enumerate}
Update $P\rightarrow P-1$ and repeat algorithm until the correct position and value for $\delta_2$ is found.  If $\delta_2$ is found but $\delta_1$ remains unknown, then apply the one-deletion algorithm D1 to determine $\delta_1$.

\begin{proof}[Proof of Algorithm D2-Binary] To prove conditions (1) and (2) are valid, we argue as follows.  \\
(1) Suppose $w_P>I$.  We consider two cases: \\
\\
(a) $x'_{P-2}=0$.  We consider two situations and show that $x'_{P-2}$ should be shifted to the right of $\delta_2$ in both situations: \\
(i) $\delta_2$ is in its correct position as the right-most deleted symbol.  In that case, since $w_P>I$, there is only one choice of symbol for $\delta_2$, namely $\delta_2=0$; otherwise, if $\delta_2=1$, then the moment for $\tilde{\mathbf{x}}$ will exceed that of $\mathbf{x}$ up to position $P$, regardless of the position of $\delta_1$ in the final decoding for $\tilde{\mathbf{x}}$:
\begin{align*}
M_P(\tilde{\mathbf{x}}) & \geq M_P(\mathbf{x}')+M(\delta_2)= M_P(\mathbf{x}')+w_P\delta_2 \\
& > M_P(\mathbf{x}')+I=M_P(\mathbf{x}).
\end{align*}
But observe that setting $\delta_2=0$ is equivalent to shifting $x'_{P-2}=0$ to the right of $\delta_2$.  Thus, we choose to shift instead. \\
(ii) $\delta_2$ is NOT in its correct position as the right-most deleted symbol.  In that case, we are forced to shift $x'_{P-2}$ to the right of $\delta_2$. \\
\\
(b) $x'_{P-2}=1$.  We consider two sub-cases: \\
(i) $I<w_P-w_{P-2}$.  We claim that $\delta_2$ is in its correct position as the right-most deleted symbol.  Otherwise, we are forced to shift $x'_{P-2}$ to the right of $\delta_2$, but then the moment of $\tilde{\mathbf{x}}$ will exceed that of $\mathbf{x}$ up to position $P$:
\begin{align*}
M_P(\tilde{\mathbf{x}}) & \geq M_P(\mathbf{x}')+(w_P-w_{P-2} ) \\
& \geq M_P(\mathbf{x})-I +(w_P-w_{P-2} )> M_P(\mathbf{x}).
\end{align*}
\noindent Thus, $\delta_2$ is in its correct position and moreover, $\delta_2=0$, since $w_P>I$. \\
(ii) $I\geq w_P-w_{P-2}$.  We claim that $\delta_2$ is NOT in its correct position.  Otherwise, $\delta_2=0$ since $w_P>I$ and so
\[
\tilde{\mathbf{x}}=(x'_1,x'_2,\dots,x'_{P-2},\delta_1,\delta_2=0,x'_{P+1},\ldots,x'_n).
\]
 But then the moment of $\tilde{\mathbf{x}}$, which is maximized if $\delta_1=1$, will always be strictly less than the moment of $\mathbf{x}$ up to position $P$:
\begin{align*}
M_P(\tilde{\mathbf{x}}) & \leq M_P(\mathbf{x}')+w_{P-1} < M_P(\mathbf{x}')+w_P - w_{P-2} \\
& < M_P(\mathbf{x}')+I=M_P(\mathbf{x}).
\end{align*}
\noindent Thus, $\delta_2$ is not in its correct position.  Therefore, $\mathbf{x}'_{P-2}$ should be shifted to the right of $\delta_2$. \\
\\
(2) Suppose $w_P< I$.  We again consider two cases: \\
\\
(a) $x'_{P-2}=0$.  We claim that $\delta_2$ is in its correct position as the right-most deleted symbol.  Otherwise, if $\delta_2$ is not in its correct position, then we are forced to shift $x'_{P-2}$ to the right of $\delta_2$, in which case 
\[
\tilde{\mathbf{x}}=(x'_1,x'_2,\ldots, x'_{P-3},\delta_1, \delta_2,x'_{P-2}=0,\ldots, x'_n).
\]
But then the moment of $\tilde{\mathbf{x}}$, which is maximized if $\delta_1=\delta_2=1$, will always be less than the moment of $\mathbf{x}$ up to position $P$:
\begin{align*}
M_P(\tilde{\mathbf{x}}) & \leq M_P(\mathbf{x}')+(w_{P-2}+w_{P-1}) \\
& \leq M_P(\mathbf{x}')+w_P< M_P(\mathbf{x}')+I = M_P(\mathbf{x}).
\end{align*}
Thus, $\delta_2$ is in its correct position.  Next, we claim that $\delta_2=1$.  Otherwise, if $\delta_2=0$, then the moment of $\tilde{\mathbf{x}}$, which is maximized if $\delta_1=1$, will always be less than the moment of $\mathbf{x}$ up to position $P$:
\begin{align*}
M_P(\tilde{\mathbf{x}}) & \leq M_P(\mathbf{x}')+w_{P-1}< M_P(\mathbf{x}')+w_P \\
& < M_P(\mathbf{x}')+I = M_P(\mathbf{x}).
\end{align*}
\\
(b) $x'_{P-2}=1$. We consider two situations and show that $x'_{P-2}$ should be shifted to the right of $\delta_2$ in both situations: \\
(i) $\delta_2$ is in its correct position.  We claim that $\delta_2=1$.  Otherwise, if $\delta_2=0$, then the moment of $\tilde{\mathbf{x}}$, which is maximized if $\delta_1=1$, will always be less than the moment of $\mathbf{x}$ up to position $P$:
\begin{align*}
M_P(\tilde{\mathbf{x}}) & \leq M_P(\mathbf{x}')+(w_{P-2}+w_{P-1}) \\
& <M_P(\mathbf{x}')+w_P < M_P(\mathbf{x}')+I = M_P(\mathbf{x}).
\end{align*}
Thus, $\delta_2=1$.  But observe that setting $\delta_2=1$ is equivalent to shifting $x'_{P-2}=1$ to the right of $\delta_2$.  Thus, we choose to shift instead. \\
(ii) $\delta_2$ is NOT in its correct position.   In that case, we are forced to shift $x'_{P-2}$ to the right of $\delta_2$. \\
This completes the proof.
\end{proof}

We now demonstrate algorithm D2-Binary in the following example to show how the problem of decoding two deletions can be reduced to that of decoding one deletion.

\begin{algorithm}[H]
 \floatname{algorithm}{Algorithm D2-Binary}
\begin{algorithmic}[1]
\State $\tilde{\mathbf{x}}=x'_1  x'_2 \ldots x'_{n-2}\delta_1\delta_2$\Comment{Initialize $\tilde{\mathbf{x}}$ by appending variable symbols $\delta_1$ and $\delta_2$ to $\mathbf{x}'$ at positions $n-1$ and $n$ respectively, where $\delta_1$ and $\delta_2$ are to be determined.}

\For {$P=n$ to $1$} \Comment {$P$ denotes position of $\delta_2$}
\For {$\sigma_1,\sigma_2 = q-1$ to $1$} \Comment{Double nested for loop}
\If {$I=\sigma_1 w_{P-1}+\sigma_2 w_P$}
\State{$\delta_1=\sigma_1$, $\delta_2=\sigma_2$}
\State{STOP}
\EndIf
\EndFor
\If{$w_P>I$}
\If{($x'_{P-2}=0$) or ($x'_{P-2}=1$ and $I\geq w_P-w_{P-2}$)}
\State $\tilde{\mathbf{x}}=x'_1 x'_2...x'_{P-3}\delta_1 \delta_2 x'_{P-2} x'_{P+1}...x'_{n-2}$ \Comment{Shift $x'_{P-2}$ to the right of $\delta_2$ and update $\tilde{\mathbf{x}}$}
\State  $I = I - x'_{P-2}(w_P-w_{P-2})$ \Comment{Update the index}
\Else
\State $\mathbf{x}'=x'_1 x'_2 ...x'_{P-2}\delta_1 0 x'_{P+1}...x'_{n-2}$ \Comment{Insert $0$ for $\delta_2$ and update $\tilde{\mathbf{x}}$}
\State {Call algorithm D1 to decode $\delta_1$}
\State{STOP}
\EndIf
\Else \Comment{$w_P<I$}
\If{$x'_{P-2}=0$}
\State $\mathbf{x}'=x'_1 x'_2...x'_{P-2}\delta_1 1 x'_{P+1}...x'_{n-2}$ \Comment{Insert $1$ for $\delta_2$ and update $\tilde{\mathbf{x}}$}
\State $I = I - w_P$ \Comment{Update the index}
\Else \Comment{$x'_{P-2}=1$}
\State $x'=x'_1 x'_2...x'_{P-3}\delta_1 \delta_2 x'_{P-2} x'_{P+1}...x'_{n-2}$ \Comment{Shift $x'_{P-2}$ to the right of $\delta_2$ and update $\tilde{\mathbf{x}}$}
\State $I = I - x'_{P-2}(w_{P}-w_{P-2})$ \Comment{Update the index}
\State {Call algorithm D1 to decode $\delta_1$}
\State{STOP}
\EndIf
\EndIf
\EndFor
\end{algorithmic}
\caption{(Decode Two Deletions)}
\end{algorithm}

\begin{example}
 Suppose $x\in C_{10}(2,2,w_{11},62)$ was transmitted and $\mathbf{x}'=(1,1,0,1,0,1,0,1)$ was received so that two deletions occurred. 
The weights $w_i$ are defined by $w_i=1+w_{i-1}+w_{i-2}$ (see Table \ref{table:weights-2}).  Therefore $m=w_{11}=232$ and $M(\mathbf{x}')=84$. 

 \begin{table}[ht]
\caption{Weights $w_i$ for $d=2,\;q=2$}
\centering
\begin{tabular}{||c| c|| c| c||}
\hline
$i$ & $w_i$ & $i$ & $w_i$ \\
\hline
1 & 1 & 7 & 33 \\
\hline
2 & 2 & 8 & 54 \\
\hline
3 & 4 & 9 & 88 \\
\hline
4 & 7 & 10 & 143 \\
\hline
5 & 12 & 11 & 232 \\
\hline
6 & 20 & 12 & 376 \\
\hline
\end{tabular}
\label{table:weights-2}
\end{table}

Since $M(\mathbf{x}')>r=62$, it follows that $M(\mathbf{x})=r+m=294$. Thus, $I=M(\mathbf{x})-M(\mathbf{x}')=210$. We initialize
 \[
 \tilde{\mathbf{x}}=(1,1,0,1,0,1,0,1, \delta_1 , \delta_2)
 \]
and apply algorithm D2.  Since (\ref{eq:test}) fails, we compare $w_{10}=143$ with $I$.  As $w_{10}<I$ and $x'_8=1$, we shift $x'_8$ to the right of $\delta_2$ and update the index: $I\rightarrow I-(w_{10}-w_8)=121$.  Then $\tilde{\mathbf{x}}$ takes the form
\[
 \tilde{\mathbf{x}}=(1,1,0,1,0,1,0, \delta_1 , \delta_2,1).
 \]
Again, since (\ref{eq:test}) fails, we compare $w_9=88$ with $I$.  As $w_9<I$ and $x'_7=0$, we set $\delta_2=1$ and update $I\rightarrow I-w_9=33$ so that
  \[
 \tilde{\mathbf{x}}=(1,1,0,1,0,1,0, \delta_1 ,1,1).
 \]
 From here, we apply algorithm D1 to determine $\delta$, which yields $\delta_1=1$ at position $7$.  Thus, 
\[
\tilde{\mathbf{x}}=(1,1,0,1,0,1,1,0,1,1)=\mathbf{x}.
\]
\end{example}

\subsection{Decoding Multiple Deletions}

Suppose $\mathbf{x}'$ is obtained from $\mathbf{x} \in C_n(q,d,m,r)$ after deleting $c$ symbols from $\mathbf{x}$, where $2\leq c\leq d$. Then to decode $\mathbf{x}'$, whose length is $n-c$, we again define $\tilde{\mathbf{x}}$ to be an initial decoding of $\mathbf{x}$ where we insert $c$ variable symbols $\delta_1, \ldots,\delta_c$ at the right end of $\mathbf{x}'$, namely at positions $P-c+1,\ldots, P$, where we initially set $P=n$:
\[
\tilde{\mathbf{x}}=(x'_1,x'_2,\ldots,x'_{n-c},\delta_1,\ldots, \delta_c).
\]
We calculate $I=M(\mathbf{x})-M(\mathbf{x}')$ (use Lemma \ref{le:appendix-3} in Appendix A to determine $M(\mathbf{x})$).  As before, our algorithm essentially determines whether to set the right-most symbol $\delta_c$ equal to an alphabet symbol ($0,\ldots, q-1$), in which case the decoding reduces to algorithm DM for $c-1$ deletions, or shift $x'_{P-c}$ (initially $x'_{n-1}$) to the right of $\delta_c$.  The following conditions describe when each action is executed.  \\
 \\
{\bf Algorithm DM (Decode Multiple Deletions):} Let $P=n$.  If 
\begin{equation} \label{eq:test-2}
I=\sigma_1w_{P-c+1}+\sigma_2 w_{P-c+2}+\ldots + \sigma_c w_P
\end{equation}
for a set of values $\sigma_1,\ldots,\sigma_c \in \{0,1,\ldots,p\}$, then $\delta_1,\ldots,\delta_c$ are in their correct positions as symbols that were deleted from $\mathbf{x}$. To decode, set $\delta_i=\sigma_i$ for $i=1,\ldots,c$.  Otherwise, we assume that at least one of the symbols $\delta_1,\ldots,\delta_c$ are NOT in their correct positions in what follows:
Define
\[
\sigma _{\max}=\max\{\sigma: \sigma(w_P-w_{P-c})<I, \sigma = 0,1,\ldots, p \}.
\]
Then
\begin{enumerate}
\item For $w_P>I$: 
\begin{enumerate}
\item If $x'_{P-c}=0$, then shift $x'_{P-c}$ to the right of $\delta_c$, update the position $P\rightarrow P-1$, and repeat algorithm.

\item If $x'_{P-c}\geq 1$ and
\begin{enumerate}
\item $I < w_P-w_{P-c}$, then set $\delta_c=0$ and apply algorithm DM on the truncated codeword $(\tilde{\mathbf{x}})_{P-1} = (x'_1,\ldots,x'_{P-c},\delta_1,\ldots,\delta_{c-1})$ with index $I$ to correct $c-1$ deletions.
\item $I\geq w_P-w_{P-c}$, then shift $x'_{P-c}$ to the right of $\delta_c$, update the index $I \rightarrow I-x'_{P-c}(w_P-w_{P-c})$, update the position $P\rightarrow P-1$, and repeat algorithm.
\end{enumerate}
\end{enumerate}

\item For $w_P < I$: 
\begin{enumerate}
\item If $x'_{P-c} > \sigma_{\max}$, then set $\delta_c=\sigma_{\max}$, update the index $I\rightarrow I-\sigma_{\max}w_P$, and apply algorithm DM on the truncated codeword $(\tilde{\mathbf{x}})_{P-1} = (x'_1,\ldots,x'_{P-c},\delta_1,\ldots,\delta_{c-1})$ with index $I$ to correct $c-1$ deletions.
\item If $x'_{P-c} < \sigma_{\max}$ and
\begin{enumerate}
\item $\sigma_{\max} w_P\leq I$, then set $\delta_c=\sigma_{\max}$, update the index $I\rightarrow I-\sigma_{\max}w_P$, and apply algorithm DM on the truncated codeword $(\tilde{\mathbf{x}})_{P-1} = (x'_1,\ldots,x'_{P-c},\delta_1,\ldots,\delta_{c-1})$ with index $I$ to correct $c-1$ deletions.
\item $\sigma_{\max} w_P > I$, then shift $x'_{P-c}$ to the right of $\delta_c$, update the position $P\rightarrow P-1$, and repeat algorithm.
\end{enumerate}
\item If $x'_{P-c} = \sigma_{\max}$, then shift $x'_{P-c}$ to the right of $\delta_c$, update the index $I \rightarrow I-\sigma_{\max}(w_P-w_{P-c})$, update the position $P\rightarrow P-1$, and repeat algorithm.
\end{enumerate}
\end{enumerate}

 \begin{proof} We prove that the conditions (1) and (2) in algorithm DM give a correct decoding of $\mathbf{x}'$. \\
(1) Suppose $w_P>I$.  We consider two cases: \\
\\
(a) $x'_{P-c}=0$.  We consider two situations and show that $x'_{P-c}$ should be shifted to the right of $\delta_c$ in both situations: \\
(i) $\delta_c$ is in its correct position as the right-most deleted symbol.  In that case, since $w_P>I$, there is only one choice of symbol for $\delta_c$, namely $\delta_c=0$; otherwise, if $\delta_c \geq 1$, then the moment for $\tilde{\mathbf{x}}$ will exceed that of $\mathbf{x}$ up to position $P$, regardless of the position and values of the other symbols $\delta_1,\ldots,\delta_{c-1}$ in the final decoding for $\tilde{\mathbf{x}}$:
\begin{align*}
M_P(\tilde{\mathbf{x}}) & \geq M_P(\mathbf{x}')+M(\delta_c)= M_P(\mathbf{x}')+w_P\delta_c \\
& > M_P(\mathbf{x}')+I=M_P(\mathbf{x}).
\end{align*}
But observe that setting $\delta_c=0$ is equivalent to shifting $x'_{P-c}=0$ to the right of $\delta_c$ (and later setting $\delta_1,\ldots, \delta_{c-1}$ equal to appropriate values determined by our algorithm).  Thus, we choose to shift instead. \\
(ii) $\delta_c$ is NOT in its correct position as the right-most deleted symbol.  In that case, we are forced to shift $x'_{P-c}$ to the right of $\delta_c$. \\
\\
(b): $x'_{P-c}\geq 1$.  We consider two sub-cases: \\
(i) $I< w_P-w_{P-c}$.  We claim that $\delta_c$ is in its correct position as the right-most deleted symbol.  Otherwise, we are forced to shift $x'_{P-c}$ to the right of $\delta_c$, but then the moment of $\mathbf{x}'$ will exceed that of $\mathbf{x}$ up to position $P$, regardless of the position and values of the other symbols $\delta_1,\ldots,\delta_{c-1}$:
\begin{align*}
M_P(\tilde{\mathbf{x}}) & \geq M_P(\mathbf{x}')+x'_{P-c}(w_P-w_{P-c} ) \\
& \geq M_P(\mathbf{x})-I +(w_P-w_{P-c} )> M_P(\mathbf{x}).
\end{align*}
\noindent Thus, $\delta_c$ is in its correct position and moreover, $\delta_c=0$ since $w_P>I$. \\
(ii) $I\geq w_P-w_{P-c}$.  We claim that $\delta_c$ is NOT in its correct position.  Otherwise, $\delta_c=0$ since $w_P>I$ and so
\[
\tilde{\mathbf{x}}=(x'_1,x'_2,\dots,x'_{P-c},\delta_1,\ldots,\delta_{c-1},0,x'_{P+1},\ldots,x'_n).
\]
But then the moment of $\tilde{\mathbf{x}}$, which is maximized if $\delta_1=\ldots=\delta_{c-1}=p$ (recall $p=q-1$), will always be strictly less than the moment of $\mathbf{x}$ up to position $P$:
\begin{align*}
M_P(\tilde{\mathbf{x}}) & \leq M_P(\mathbf{x}')+p(w_{P-c+1}+\ldots + w_{P-1}) \\
& < M_P(\mathbf{x}')+w_P - p (w_{P-c} + \ldots + w_{P-d})  \\
& < M_P(\mathbf{x})+w_P-w_{P-c}  \\
& \leq M_P(\mathbf{x}')+I=M_P(\mathbf{x}).
\end{align*}
\noindent Thus, $\delta_c$ is not in its correct position.  Therefore, $\mathbf{x}'_{P-c}$ should be shifted to the right of $\delta_c$. \\
\\
(2) Suppose $w_P< I$.  We first prove that if $\delta_c$ is in its correct position, then $\delta_c=\sigma_{\max}$.  We rule out all other possible values as follows: \\
(i) Suppose $\delta_c = \sigma < \sigma_{\max}$.  But then the moment of $\tilde{\mathbf{x}}$ up to position $P$, which is maximized if $\delta_1=\ldots=\delta_{c-1}=p$, will always be less than the moment of $\mathbf{x}$ because of the following calculation (recall the recurrence satisfied by $w_P$ and the fact that $\sigma_{\max}(w_P-w_{P-c}) < I$):
\begin{align*}
& M_P(\tilde{\mathbf{x}}) \\
& \leq M_P(\mathbf{x}')+p(w_{P-c+1}+\ldots + w_{P-1}) +\sigma\cdot w_P \\
& < M_P(\mathbf{x}')+(\sigma+1) w_P - p (w_{P-c}+\ldots + w_{P-d}) \\
& < M_P(\mathbf{x})+\sigma_{\max}(w_P-w_{P-c}) \\
& < M_P(\mathbf{x}')+I  =M_P(\mathbf{x}).
\end{align*}
(ii) Suppose $\delta_c = \sigma > \sigma_{\max}$.  But then the moment of $\tilde{\mathbf{x}}$ up to position $P$, which is minimized if $\delta_1=\ldots=\delta_{c-1}=0$, will always be greater than the moment of $\mathbf{x}$ because of a similar calculation:
\begin{align*}
M_P(\tilde{\mathbf{x}}) & \geq M_P(\mathbf{x}')+\sigma\cdot w_P \\
& > M_P(\mathbf{x}')+\sigma(w_P - w_{P-c}) \\
& \geq M_P(\mathbf{x}')+I =M_P(\mathbf{x}).
\end{align*}
Thus, $\delta_c=\sigma_{\max}$ if it is in its correct position.

Next, we consider three cases: \\
\\
(a) $x'_{P-c} > \sigma_{\max}$.  We claim that $\delta_c$ is in its correct position.  Otherwise, we are forced to shift $x'_{P-c}$ to the right of $\delta_c$, but then the moment of $\tilde{\mathbf{x}}$, which is minimized if $\delta_1=\ldots=\delta_{c-1}=0$, will always be greater than the moment of $\mathbf{x}$ up to position $P$:
\begin{align*}
M_P(\tilde{\mathbf{x}}) & \geq M_P(\mathbf{x}')+x'_{P-c}(w_P - w_{P-c}) \\
& > M_P(\mathbf{x}')+I  =M_P(\mathbf{x}).
\end{align*}
Thus, $\delta_c$ is in its correct position and as we argued previously, $\delta_c=\sigma_{\max}$. \\
\\
(b) $x'_{P-c}< \sigma_{\max}$.  We consider two sub-cases: \\
(i) $\sigma_{\max} w_P\leq I$.  We claim that $\delta_c$ is in its correct position.  Otherwise, we are forced to shift $x'_{P-c}$ to the right of $\delta_c$, but then the moment of $\tilde{\mathbf{x}}$, which is maximized if $\delta_1=\ldots=\delta_{c-1}=p$, will always be less than the moment of $\mathbf{x}$ up to position $P$.  This is because
\begin{align*}
M_P(\tilde{\mathbf{x}}) & \leq M_P(\mathbf{x}')+p(w_{P-c}+\ldots + w_{P-1}) \\
& \ \ \ \ +x'_{P-c} (w_P-w_{P-c}) \\
& < M_P(\mathbf{x}')+ w_P -p(w_{P-c-1}\ldots+ w_{P-d})  \\
& \ \ \ \ + x'_{P-c} (w_P-w_{P-c})  \\
& < M_P(\mathbf{x}')+(1+x'_{P-c}) w_P \\
\end{align*}
Next, we use the fact that $x'_{P-c}< \sigma_{\max}$ to obtain
\begin{align*}
M_P(\tilde{\mathbf{x}}) & \leq M_P(\mathbf{x}')+\sigma_{\max} w_P \\
& \leq M_P(\mathbf{x}')+I  =M_P(\mathbf{x}).
\end{align*} 
Thus, we set $\delta_c=\sigma_{\max}$.\\
\\
(ii) $\sigma_{\max}w_P > I$.  We claim that $\delta_c$ is NOT in its correct position.  Otherwise, if $\delta_c$ is in its correct position, then we must have $\delta_c= \sigma_{\max}$ and so the moment of $\tilde{\mathbf{x}}$  up to position $P$, which is minimized if $\delta_1=\ldots=\delta_{c-1}=0$, will always be greater than the moment of $\mathbf{x}$:
\begin{align*}
M_P(\tilde{\mathbf{x}}) & \geq M_P(\mathbf{x}')+\sigma_{\max}w_P \\
& > M_P(\mathbf{x}')+I =M_P(\mathbf{x}).
\end{align*}
Thus, we shift $x'_{P-c}$ to the right of $\delta_c$. \\
\\
(c) $x'_{P-c}=\sigma_{\max}$.  In this case, observe that if $\delta_c$ is in its correct position, then $\delta_c =\sigma_{\max}$, but this same result can be achieved by shifting $x'_{P-c}$ to the right of $\delta_c$ (and later setting $\delta_1,\ldots, \delta_{c-1}$ equal to appropriate values determined by our algorithm).  Thus, we choose to shift instead.  This completes the proof.
\end{proof}

 We demonstrate algorithm DM with the following example.

\begin{example} Suppose a ternary codeword $\mathbf{x}\in C_8(3,2,w_9,23)$ of length 8 was transmitted and the deleted codeword received $\mathbf{x}'=(1,2,2,0,1,2)$ was received where two symbols were deleted.  We have $m=w_{9} = 3861$; see Table \ref{table:weights} for a list of the weights $w_i$.
 
Since $M(\mathbf{x}') =504 > r$, it follows from Lemma \ref{le:appendix-3} that $M(\mathbf{x}) = m+r = 3884$.  Thus, the index $I = M(\mathbf{x})-M(\mathbf{x}')=3380$. We now apply algorithm DM by defining our initial decoding as
\[
\tilde{\mathbf{x}}=(1,2,0,2,1,2,\delta_1,\delta_2).
\]
Since (\ref{eq:test-2}) fails, we compare $w_8=1413$ with $I$.  As $w_8 < I$, we compute $\sigma_{\max}=2$.  Since $x'_6=2=\sigma_{\max}$, we shift $x'_6$ to the right of $\delta_2$ and update $I\rightarrow I-x'_6(w_8-w_6)=932$ so that
\[
\tilde{\mathbf{x}}=(1,2,0,2,1,\delta_1,\delta_2,2).
\]
Again, since (\ref{eq:test-2}) fails, we compare $w_7=517$ with $I$.  As $w_7<I$, we calculate $\sigma_{\max}=2$.  Since $x'_5=1<\sigma_{\max}$ and $\sigma_{\max} w_7 =1034 > I$, we shift $x'_5$ to the right of $\delta_2$ and update $I\rightarrow I-x'_5(w_7-w_5)=484$ so that
\[
\tilde{\mathbf{x}}=(1,2,0,2,\delta_1,\delta_2,1,2).
\]
Since (\ref{eq:test-2}) fails again, we compare $w_6=189$ with $I$.  As $w_6 < I$, we calculate $\sigma_{\max}=2$.  Since $x'_4=2=\sigma_{\max}$, we shift $x'_6$ to the right of $\delta_2$ and update $I\rightarrow I - x'_4(w_6-w_4)=156$ so that
\[
\tilde{\mathbf{x}}=(1,2,0,\delta_1,\delta_2,2,1,2).
\]
Again, since (\ref{eq:test-2}) fails, we compare $w_5=69$ with $I$.  As $w_5< I$, we calculate $\sigma_{\max}=2$.  Since $x'_3=0 < \sigma_{\max}$ and $\sigma_{\max} w_5 =138 < I$, we set $\delta_2=\sigma_{\max}=2$ and update the index $I\rightarrow I-\sigma_{\max} w_5 = 18$.  This yields
\[
\tilde{\mathbf{x}}=(1,2,0,\delta_1,2,2,1,2).
\]

It remains to apply algorithm D1 on the truncated codeword $(\tilde{\mathbf{x}})_4=(1,2,0,\delta_1)$ with $I=18$ to decode $\delta_1$.  Following Example \ref{ex:one-deletion}, we find that $\delta_1=2$ should be inserted at position 3.  Hence, our final decoding is
\[
\tilde{\mathbf{x}}=(1,2,2,0,2,2,1,2)=\mathbf{x}.
\]
\end{example}

\section{Appendix}

\subsection{Useful Lemmas}
In this appendix, we aim to show that the moment of a codeword is strictly less than twice the modulus defining its codebook.  This allows us to precisely determine its moment based on the moment of the deleted codeword.

\begin{lemma} \label{le:appendix-1}
 For $d \geq 2$,
\begin{equation}
\label{eq:summation}
\sum\limits_{i=1}^n w_i = \frac {p \left (\sum\limits_{i=0}^{d-1} (d-i)w_{n-i} \right ) -n}{pd-1}.
\end{equation}
\end{lemma}
\begin{proof} We argue by induction on $n$.  It is straightforward to verify that (\ref{eq:summation}) holds for $n=1$.  Next, assume that (\ref{eq:summation}) holds for arbitrary $n$.
Then for $n+1$, since
\begin{align*}
 \sum\limits_{i=1}^{n+1} w_i  =\sum\limits_{i=1}^n w_i + w_{n+1},
 \end{align*}
 it follows from the inductive hypothesis and the recurrence for $w_{n+1}$ that
 \begin{align*}
 &  \sum\limits_{i=1}^{n+1} w_i  \\
 & =\frac{p \left ( \sum\limits_{i=0}^{d-1} (d-i)w_{n-i} \right ) - n}{pd-1} + \frac{pd(w_{n+1})-w_{n+1}}{pd-1} \\
 & =\frac { p \left ( \sum\limits_{i=0}^{d-2} (d-i-1)w_{n-i}\right ) +pd(w_{n+1})-(n+1)}{pd-1}, \\
\end{align*}
Then re-index the summation on the right-hand side and simplifying yields
\begin{align*}
 \sum\limits_{i=1}^{n+1} w_i  = \frac{p \left ( \sum\limits_{i=0}^{d-1} (d-i)w_{n+1-i} \right ) - (n+1)}{pd-1}.
\end{align*}
Hence, (\ref{eq:summation}) holds for $n+1$.
\end{proof}

\begin{lemma} \label{le:appendix-2}
 For $d \geq 2$,
\begin{equation}
\label{eq:bound-2}
\sum\limits_{i=1}^n w_i < \frac {d}{pd-1} w_{n+1}.
\end{equation}
\end{lemma}

\begin{proof}
 It follows from Lemma \ref{le:appendix-1} that
\begin{align*}
\sum\limits_{i=1}^n w_i & = \frac {p \left ( \sum\limits_{i=0}^{d-1} (d-i) w_{n-i} \right ) -n}{pd-1} \\
& < \frac {pd \left ( \sum\limits_{i=0}^{d-1} w_{n-i} \right ) -n}{pd-1} = \frac {d(w_{n+1} -1) -n}{pd-1}
\\
& < \frac {d}{pd-1}w_{n+1}.
\end{align*}
This proves (\ref{eq:bound-2}).
\end{proof}

\begin{lemma} \label{le:appendix-3} Let $\mathbf{x}\in C_n(q,d,m,r)$.  Suppose $\mathbf{x}'$ is obtained by deleting $c$ symbols from $\mathbf{x}$, where $c\leq d$.
If $M(\mathbf{x}') >r$, then $M(\mathbf{x})=r+m$. Otherwise, if $M(\mathbf{x}')\leq r$, then $M(\mathbf{x})=r$.
\end{lemma}

\begin{proof}
Recall from our definition of $\mathbf{x}$ that $M(\mathbf{x}) \equiv r \pmod{m}$.  We claim that 
\begin{equation}
\label{eq:bound}
M(x) < 2m.
\end{equation}
This follows from Lemma 
\ref{le:appendix-2}:
\[
M(\mathbf{x})\leq p \sum\limits_{i=1}^n w_i < \frac{ps}{ps-1}w_{n+1} \leq 2w_{n+1} \leq 2m,
\]
where we have used the fact that $ps/(ps-1)\leq 2$ since $d\geq 2$ and $p\geq 1$.  

If $M(\mathbf{x}')>r$, then $M(\mathbf{x})>r$ since $M(\mathbf{x})\geq M(\mathbf{x}')$.  It follows from (\ref{eq:bound}) that $M(\mathbf{x})=r+m$.
On the other hand, if $M(\mathbf{x}')\leq r$, then we claim that $M(\mathbf{x})=r$.  To prove this, assume on the contrary that $M(\mathbf{x})=r+m$.  Then
\begin{align*}
M(\mathbf{x})-M(\mathbf{x}') &
 \leq p(w_{n-c+1}+\ldots w_{n}) \\
 & < w_{n+1}-p(w_{n-d+1}+\ldots+w_{n-c})  \\
 & < m.
\end{align*}
It follows that
\[
M(\mathbf{x}')>M(\mathbf{x})-m = r,
\]
which is a contradiction.
\end{proof}

\subsection{Sizes of Generalized Helberg Codes}

We present values for the size of the largest code in terms of the codeword length.  Given positive integers $q$, $d$, $n$ $r$, and $m=w_{n+1}$, we denote the size of the largest code $C_n(q,d,w_{n+1},r)$ by 
\[
N_n(q,d)=\max_{ r=0,1,...,w_{n+1}-1}\{|C_n(q,d,w_{n+1},r)|\}.
\]
Also, let $R_n(q,d)$ denote the set of values $r$ for which $|C_n(q,d,w_{n+1},r)|=N_n(q,d)$.
 
Through exhaustive computer search, we computed the values of $N_n(q,d)$ and $R_n(q,d)$ for certain values of $q$, $d$, and $n$.  Table \ref{table:maximum-size-binary} gives values for $N_n(2,2)$ and $R_n(2,2)$ for binary 2-deletion codes ($q=2$, $d=2$) with $n$ ranging from 1 to 15.  Tables  \ref{table:maximum-size-ternary} and  \ref{table:maximum-size-quaternary} give values for ternary 2-deletion codes ($q=3$, $d=2$) and quaternary 2-deletion codes ($q=4$, $d=2$), respectively, but over a shorter range for $n$.

\begin{table}[t]
\caption{Binary 2-Deletion Codes: Values of $N_n(2,2)$ and $R_n(2,2)$}
\centering

\begin{tabular}{|c|c|l|}
\hline
$n$ & $N_n(2,2)$ & $R_n(2,2)$  \\ [0.5ex] 
\hline
1 & 1 & 0, 1 \\ \hline
2 & 1 & 0, 1, 2, 3  \\ \hline
3 & 2 & 0 \\ \hline
4 & 2 & 0, 1, 2, 7  \\ \hline
5 & 2 & 0, 1, 2, 3, 4, 5, 6, \\
& & 
7, 12, 13, 14, 19 \\ \hline
6 & 3 & 0, 1, 6, 7, 12, 13  \\ \hline
7 & 4 & 12, 13 \\ \hline
8 & 5 & 12, 33 \\ \hline
9 & 6 & 12, 33, 39, 45, 66  \\ \hline
10 & 8 & 66 \\ \hline
11 & 9 & 65, 66, 99, 100, 120, \\
& & 
121, 154, 155 \\ \hline
12 & 11 & 65, 66, 99, 154, 155, 175, \\
& & 176, 181, 182, 187, 188, \\
& & 208, 209, 264, 297, 298 \\ \hline
13 & 15 & 297, 298 \\ \hline
14 & 18  & 297, 441, 475, 496, 530, 674 \\ \hline
15 & 22 & 297, 441, 674, 763, 784, 790, \\
& & 796, 817, 906, 1139, 1283  \\ \hline
16 & 30 & 1283 \\ \hline
\end{tabular}

\label{table:maximum-size-binary}
\end{table}

\begin{table}[t]
\caption{Ternary 2-Deletion Codes:  Values of $N_n(3,2)$ and $R_n(3,2)$}
\centering

\begin{tabular}{|c|c|l|}
\hline
$n$ & $N_n(3,2)$ & $R_n(3,2)$  \\ [0.5ex] 
\hline
1 & 1  & 0, 1, 2 \\ \hline
2 & 1  & 0, 1, 2, 3, 4, 5, 6, 7, 8 \\ \hline
3 & 2  & 0, 1 \\ \hline
4 & 2  & 0, 1, 2, 3, 4, 5, 6, 7, 25, \\
& & 26, 50, 51 \\ \hline
5 & 3  & 0, 25 \\ \hline
6 & 4 & 25, 50 \\ \hline
7 & 4 & 24, 25, 50, 69, 70, 71, 72, 73, 74, \\
& & 75, 94, 119, 138, 139, 140, 141,\\
& & 142,143, 144, 163, 188,189, 542, \\
& & 567, 1059, 1084
 \\ \hline
8 & 5 & 24, 25, 49, 50, 69, 70, 71, 72, 73, \\
 & & 74, 188, 189, 213, 214, 377, 378, \\
& & 402, 403, 517, 518, 519, 520, \\
& & 521, 522, 541, 542, 566, 567 \\ \hline
9 & 7  & 541, 542, 566, 567, 1058, 1059, \\
 & & 1083, 1084   \\ \hline
10 & 8  & 517, 518, 519, 520, 521, 541, 542, \\
& & 566, 567, 1437, 1482, 1483, 1484, \\
& & 1485, 1486, 1487, 1551, 1552, \\
& & 1553, 1554, 1555, 1556, 1601, \\
 & & 2850, 2895, 2896, 2897, 2898, \\
 & & 2899, 2900, 2964, 2965, 2966, \\
 & & 2967, 2968, 2969, 3014, 3884, \\
 & & 3885, 3909, 3910, 3930, 3931, \\
 & & 3932, 3933, 3934  \\ \hline
\end{tabular}

\label{table:maximum-size-ternary}
\end{table}

\begin{table}[t]
\caption{Quaternary 2-Deletion Codes: Values of $N_n(4,2)$ and $R_n(4,2)$}
\centering

\begin{tabular}{|c|c|l|}
\hline
$n$ & $N_n(4,2)$ & $R_n(4,2)$  \\ [0.5ex] 
\hline
1 & 1  & 0, 1, 2, 3 \\ \hline
2 & 1  & 0, 1, 2, 3, 4, 5, 6, 7, 8, 9, 10, 11, \\
 & & 12, 13, 14, 15 \\ \hline
3 & 2  & 0, 1, 2 \\ \hline
4 & 2  & 0, 1, 2, 3, 4, 5, 6, 7, 8, 9, 10, 11, \\
& & 12, 13, 14, 61, 62, 63, 122, 123, \\
& & 124, 183, 184, 185 \\ \hline
5 & 3  & 0, 1, 61, 62 \\ \hline
6 & 4  & 61, 62, 122, 123, 183, 184 \\ \hline
7 & 5  & 61, 880  \\ \hline
8 & 6  & 61, 122, 183, 880, 941, 1760, \\
& & 1821, 2640, 2701, 3398, \\
& & 3459, 3520   \\ \hline
\end{tabular}

\label{table:maximum-size-quaternary}
\end{table}


\end{document}